\documentclass[acmlarge,nonacm]{acmart}

\usepackage{amsmath}

\usepackage{thmtools}
\usepackage{thm-restate}

\usepackage{algorithm}
\usepackage{algpseudocode}
\usepackage{indentfirst}

\usepackage{cleveref}

\usepackage{graphicx}
\usepackage{subcaption}

\usepackage[normalem]{ulem}
\usepackage{color}

\declaretheorem[name=Definition, numbered=no]{definition*}

\AtBeginDocument{}

\AtEndPreamble{%
    \theoremstyle{acmdefinition}
    
}

\newcommand{\Oh}{\mathcal{O}}
\newcommand{\oh}{o}

\newcommand{\Nat}{\mathbb{N}}

\newcommand{\Ps}{\mathcal{P}}

\newcommand{\HModel}{H(n; \mathcal{M})}
\newcommand{\HModelExpected}{\widehat{H}(n; \mathcal{P})}
\newcommand{\HUniform}{H_k(n, M)}
\newcommand{\HUniformExpected}{H_k(n, p)}
\newcommand{\RandomCSP}{\widehat{CSP}_{n, p}(\mathcal{Q})}

\newcommand{\Prob}[1]{\mathbb{P}\left[#1\right]}
\newcommand{\EV}[1]{\mathbb{E}\left[#1\right]}

\begin{document}
    \title{Property B for random non-uniform hypergraphs}

  
    \author{Grzegorz Ryn}
    \additionalaffiliation{%
        \department{Doctoral School of Exact and Natural Sciences}
        \institution{Jagiellonian University}
        \city{Kraków}
        \country{Poland}
    }

    \author{Jakub Kozik}
    \affiliation{%
        \institution{Theoretical Computer Science, Jagiellonian University}
        \city{Kraków}
        \country{Poland}
    }

    \begin{abstract}
        We show that, in random hypergraphs with several permitted edge sizes, non-uniformity affects existential and algorithmic bounds for 2-colorability (Property B) in fundamentally different ways. Assigning weight $2^{-k}$ to each $k$-edge, we obtain upper and lower bounds in terms of the total edge weight that asymptotically match the uniform bounds as the minimum permitted edge size grows, regardless of how edges are distributed among sizes. We also extend the best-known algorithm for 2-coloring random uniform hypergraphs to the non-uniform setting. With weight $\frac{k}{2^k}$ assigned to each $k$-edge, we construct non-uniform instances whose expected total edge weight per vertex is arbitrarily large, yet the algorithm finds a proper coloring asymptotically almost surely. In the uniform setting, the algorithm fails with high probability once this quantity exceeds a constant. Our construction uses sufficiently separated edge sizes, so that edges of different sizes become relevant at well-separated stages of the execution and their effects are essentially independent.
    \end{abstract}



    \maketitle

    \section{Introduction}
        A hypergraph $\mathcal{H} = \left( V, E \right)$ consists of a vertex set $V$ and a family $E \subseteq \mathcal{P}(V)$ of edges. A 2-coloring $c: V \rightarrow \{\text{RED},\text{BLUE}\}$ is \emph{proper} if no edge $e \in E$ is \emph{monochromatic}, i.e., every edge contains both colors. A hypergraph admitting such a coloring is called $2$-colorable, or is said to have \emph{Property B}. Equivalently, Property B is the satisfiability problem for a family of \emph{not-all-equal} constraints on Boolean variables.

        So far, random instances of Property B have been studied primarily in the uniform case, where every edge has the same size $k$. There, the existential question is fairly well understood. Almost all random hypergraphs with slightly fewer than $\frac{\ln 2}{2} \cdot 2^k \cdot n$ edges are $2$-colorable, while above that density proper colorings cease to exist. The best known algorithm, due to Achlioptas, Kim, Krivelevich, and Tetali~\cite{AKKT02}, which we abbreviate to AKKT, reaches only a density smaller by a factor of order $k$. In this paper we study what happens when edges of several different sizes are allowed to occur at once. As it turns out, the two classical aspects of the problem, namely the existence of a proper coloring and our ability to find one efficiently, react to such a change in fundamentally different ways.

        Larger edges are less likely to be monochromatic under a uniformly random $2$-coloring, and the edge count alone is therefore too coarse to measure non-uniform hypergraphs. Erdős and Lovász~\cite{EL1975} addressed this in the worst-case problem by assigning the weight $2^{1-|e|}$ to each edge $e$, so that the total weight of a hypergraph is precisely the expected number of monochromatic edges. We follow the same idea
assigning weight $2^{-k}$ to a $k$-edge in the existential setting and $k2^{-k}$ in the algorithmic one. Under such normalization, each of the uniform bounds corresponds to a total edge weight proportional to $n$, and we refer to the expected total weight divided by $n$ as the \emph{normalized density}.

        In the existential setting, this single quantity appears to capture most of what matters. Both of our bounds depend on the instance only through its normalized density, with every permitted edge size contributing additively, and as the minimum permitted size grows they approach the constant $\frac{\ln 2}{2}$ known from the uniform model. In this sense, mixing edge sizes seems to bring neither benefit nor penalty, and the uniform bounds survive non-uniformity to leading order.

        For AKKT the situation is different, and no such quantity seems to suffice. Edges of different sizes obstruct the algorithm at different \emph{moments} of its execution, with larger edges becoming dangerous later, so that sizes which are sufficiently far apart never obstruct it simultaneously. By assigning to each of them as much weight as it can carry on its own, we obtain instances of arbitrarily large normalized density on which AKKT still succeeds. Since the algorithm is known to fail on uniform instances of large normalized density, the same value of this parameter may correspond to failure in the uniform case and to success in a non-uniform one. In particular, no condition depending on the normalized density alone can characterize the success of AKKT.

        \subsection{Property B and the uniform benchmark}
        The foundational work on the 2-colorability of such structures was done by Bernstein \cite{Ber08} over a century ago. For this historical reason, hypergraph 2-colorability is often called \emph{Property B}. A significant number of results focus on $k$-uniform hypergraphs ($k$-graphs for short), where every edge consists of exactly $k$ vertices. As shown by Lovász \cite{Lovasz73}, this restriction still makes the topic interesting, as for $k \geq 3$ the computational problem of recognizing 2-colorable $k$-graphs is NP-complete.

        In the worst-case setting, the effect of non-uniformity has already been studied through the weighting recalled above. A $k^+$-graph is a hypergraph in which every edge has size at least $k$. Erdős and Lovász asked how small the total weight of a non-$2$-colorable $k^+$-graph can be, and conjectured that this minimum tends to infinity with $k$, which was confirmed by Beck~\cite{Beck78}; the best known lower bound is $\Omega(\log k)$~\cite{DGK18}. While it seems harder to color non-uniform hypergraphs, no known construction of a non-$2$-colorable hypergraph exploits non-uniformity. In fact, the best upper bound remains $\Oh(k^2)$, which follows from Erdős' probabilistic argument for the uniform case, recalled below.

        Fixing the constraint size $k$ and letting the number of vertices $n$ tend to infinity leads to direct analogues of the well-known random graph models $G(n,p)$ and $G(n, M)$. Let $\HUniform$ and $\HUniformExpected$ denote two random $k$-graph models. In $\HUniform$, exactly $M$ $k$-edges are chosen uniformly at random, whereas in $\HUniformExpected$ each of the $\binom{n}{k}$ possible $k$-edges is included independently with probability $p$.
We consider sequences of these models in which $n$ tends to infinity and the parameter $M$ or $p$ depends on $n$.
 An event holds \emph{asymptotically almost surely} (a.a.s.) along such a sequence if its probability tends to one as $n\to\infty$. Studying $2$-colorability in these models shifts the focus from the worst case to the \emph{typical} behaviour under different edge-density regimes. In particular, it allows us to ask when proper $2$-colorings exist or when they can be found efficiently, while avoiding problematic but improbable instances.

        We recall three benchmarks from the uniform model. The first is the density scale on which proper colorings cease to exist. In the early 1960s, Erdős and Hajnal~\cite{ErdHaj61} asked about the minimum number of edges in a non-$2$-colorable $k$-graph. Shortly after, in 1964, Erdős~\cite{Erd1964} showed with a probabilistic argument that there are non-$2$-colorable $k$-graphs with $\Oh(k^2 \cdot 2^k)$ edges. This argument can be directly recreated in the $\HUniform$ model, yielding the following. See Alon and Spencer~\cite{ASunpub} for the adapted version of the proof.

        \begin{restatable}[Non-constructive upper bound for $k$-graphs \cite{ASunpub}]{thm}{uniformUpperBound} 
            \label{thm:uniformUpperBound}
            \hfill \\
            For $M \geq \frac{\ln 2}{2} \cdot 2^k \cdot n$, a random hypergraph drawn from $\HUniform$ is a.a.s. non-$2$-colorable.
        \end{restatable}
        The asymptotic is made with $k$ fixed and $n \to \infty$, but similar argument is also valid for $k$ increasing with $n$. For details, see \cite{DKS21}.

        The second benchmark concerns the existence of proper colorings below this density. Alon and Spencer~\cite{ASunpub} showed that almost all $k$-graphs with at most $c\frac{2^k}{k^2}n$ edges are $2$-colorable. In 2002, Achlioptas and Moore~\cite{AM2002} improved the lower bound to asymptotically almost match the upper bound; however, their argument is not constructive.

        \begin{restatable}[Non-constructive lower bound for $k$-graphs \cite{AM2002}]{thm}{uniformLowerBound} 
            \label{thm:uniformLowerBound}
            \hfill \\
            Let $\varepsilon > 0$ and $k \geq k_0(\varepsilon)$. A random hypergraph drawn from $\HUniformExpected$, where $p \cdot \binom{n}{k} \leq \left( \frac{\ln 2}{2} \cdot \left(2^k - 1 \right) - \frac{1 + \varepsilon}{2} \right) \cdot n$ is $2$-colorable a.a.s.
        \end{restatable}
        \noindent Note that lower-order terms in $k$ have been further improved by Coja-Oghlan and Zdeborová \cite{CZ2012}.

        The third benchmark is the algorithmic scale. In 2002, AKKT~\cite{AKKT02} gave an algorithm that a.a.s. $2$-colors random $k$-graphs drawn from $\HUniformExpected$, provided that the expected number of edges is at most $c\frac{2^k}{k}n$. Unlike the nearly tight existential bounds, this guarantee is smaller by a factor of order $k$, and it remains the state of the art. The resulting gap between the existence of proper colorings and our ability to find them efficiently is what makes the algorithmic question interesting. The work of Achlioptas and Coja-Oghlan~\cite{ACO08} provides evidence that this gap is not merely an artifact of the analysis; we return to it in the closing section.

        \subsection{Random non-uniform models}
        We work with the natural generalizations of the uniform models, in which every arity carries its own parameter. In the \emph{fixed-edge} model, a prescribed number of edges of each size is chosen uniformly at random; in the \emph{independent-edge} model, each potential edge is included independently with a probability depending on its size.

        Both are parametrized by a \emph{density profile}, a fixed, finitely supported sequence recording how much each permitted size contributes to the normalized density. The profile does not vary with $n$, so the largest permitted edge size stays bounded and every edge still represents a \emph{local} constraint, in the spirit of random CSP. Section~\ref{sec:models} makes the models and this parametrization precise, and explains how our theorems should be read: one first fixes a profile, and only then lets $n$ grow.

        \subsection{Main results}
        All colorability and algorithmic statements below hold a.a.s., with the density profile fixed and $n \to \infty$.

        For the existential question we obtain two asymptotically almost matching bounds. Theorem~\ref{thm:nonUniformUpperBound} states that a normalized density exceeding $\frac{\ln 2}{2}$ implies non-$2$-colorability, regardless of how the edges are distributed among the permitted sizes. Conversely, Theorem~\ref{thm:nonUniformLowerBound} asserts that for every fixed margin below $\frac{\ln 2}{2}$, a proper coloring exists, provided that all permitted sizes are sufficiently large. In both conditions each size contributes through its own coefficient, and the corrections separating them from $\frac{\ln 2}{2}$ vanish as the minimum permitted size grows, so the two bounds meet at the constant known from the uniform model. The upper bound follows from a first-moment argument. The lower bound requires more care, as the second-moment method alone yields only probability bounded away from zero. We upgrade this conclusion using a sharp-threshold argument based on the random CSP framework of Hatami and Molloy~\cite{HM08}.

        The algorithmic side is governed by a quantity of a different kind. AKKT colors two vertices in each round, and the main obstacle to its success is the appearance of partially colored edges which are close to becoming monochromatic. Theorem~\ref{thm:bounded3Tails} states that the algorithm succeeds as long as the expected number of such configurations created in a single round stays below $1$ by a constant margin throughout the execution. This expectation varies as the execution progresses, so unlike the existential bounds, which constrain a single number computed from the whole instance, the condition has to be verified at every stage. Theorem~\ref{thm:AKKTsuccess} expresses it explicitly in terms of the density profile.

        The consequences of this difference become apparent once we observe that each edge size contributes to this expectation within its own window of time, which occurs later and becomes narrower as the size grows. Sizes chosen sufficiently far apart may therefore be loaded almost to capacity without ever competing, which is how Theorem~\ref{thm:UnboundedWeight} provides profiles of arbitrarily large normalized density on which AKKT still succeeds.

    \section{Random non-uniform hypergraphs}
    \label{sec:models}
        
        In the random models considered below, all permitted edge sizes are at least $3$.

        \subsection{Fixed-edge and independent-edge models}
        \begin{definition*}{$H(n;\mathcal M_n)$ model} \hfill \\
            Let $\mathcal M_n$ specify an integer number of edges for each arity, with $0\leq\mathcal M_n(k)\leq\binom nk$. Then $H(n;\mathcal M_n)$ is a distribution of random hypergraphs on a vertex set $V=\{1,\ldots,n\}$ with exactly $\mathcal M_n(k)$ distinct $k$-edges selected uniformly at random. Sets of edges of different arities are chosen independently.
        \end{definition*}

        The distribution $\widehat H(n;\mathcal P_n)$ is defined analogously. Given $\mathcal P_n:\Nat\rightarrow[0,1]$, a hypergraph is sampled by including each potential nonempty edge $e$ independently with probability $\mathcal P_n(|e|)$. In practice, however, it is more convenient to parametrize this model by the expected number of edges for each arity, namely $\binom nk\mathcal P_n(k)$.

        We suppress the subscript $n$ in $\mathcal M$ and $\mathcal P$ when their dependence on $n$ is clear. For brevity, we occasionally abbreviate ``hypergraph'' to ``graph''.

        \subsection{Density profiles}
        Throughout the paper, we parametrize families of these models by a \emph{density profile}, that is, a fixed, finitely supported sequence $(\lambda_k)_{k\geq1}$ of nonnegative real numbers. Each coefficient specifies the contribution of arity $k$ to the normalized density, or a bound on this contribution. To read our theorems, first fix a profile satisfying the stated conditions. For each $n$, this profile then determines the parameters $\mathcal M_n$ or $\mathcal P_n$, or bounds on them, according to the weighting used in that theorem. The conclusion concerns this family of random hypergraphs as $n\to\infty$, with the profile kept fixed.

        With the base weighting, used in existential proofs, a coefficient $\lambda_k$ corresponds to an edge count on the scale
        \[
            \lambda_k\cdot2^k\cdot n,
        \]
        whereas under the weighting used in the algorithmic analysis, it corresponds to an expected edge count on the scale
        \[
            \lambda_k\cdot\frac{2^k}{k}\cdot n.
        \]
        For exact counts or expectations of these forms, the normalized density is $\sum_k\lambda_k$. We use the symbol $\lambda_k$ in both settings, always with the weighting determined by the parametrization in the statement. In the fixed-edge model, the displayed expressions serve as upper or lower bounds on the integer edge counts, as specified by the inequalities in each theorem. This convention avoids unnecessary floor and ceiling notation.

        \subsection{Transfer between the two models}
        The two models are closely related, just as $G(n,M)$ and $G(n,p)$ are. In our regime, standard comparisons allow results on $2$-colorability to be transferred between them, so we choose the model according to the convenience of the argument.

        We use $\HModel$ for the existential bounds, where fixed edge counts simplify the moment calculations, and $\HModelExpected$ for the algorithmic analysis, where independence allows us to expose the instance as the algorithm runs.

    \section{Additivity in the existential problem}
        As discussed in the introduction, the existence of proper $2$-colorings in $k$-graphs is much better understood than its algorithmic counterpart, with the lower and upper bounds being asymptotically almost tight. We therefore begin with the existential question, extending both uniform bounds to profiles with several permitted edge sizes. In both of them the permitted sizes contribute separately and additively, which suggests that non-uniformity plays only a minor role here.

        \subsection{First-moment upper bound}
            Alon and Spencer \cite{ASunpub} obtained their upper bound with a first-moment argument. They showed that with sufficiently many random edges, the expected number of proper $2$-colorings tends to zero, and therefore almost all such graphs are non-$2$-colorable. With the existential weighting, the uniform upper bound says that a random $k$-graph with normalized density exceeding $\frac{\ln 2}{2}$ is a.a.s. non-$2$-colorable. We show that this condition holds in the non-uniform setting as well.
            
            \begin{restatable}[Non-uniform upper bound for 2-colorability]{thm}{nonUniformUpperBound}
                \label{thm:nonUniformUpperBound}
                \hfill \\
                Let $\HModel$ be a hypergraph distribution such that $\mathcal{M}(k) \geq \lambda_k \cdot 2^k \cdot n$ for every $k$. If
                \[
                    \sum_k \left( 1 + 2^{-k} \right) \lambda_k > \frac{\ln 2}{2}
                \]
                then a hypergraph drawn from $\HModel$ is non-$2$-colorable a.a.s.
            \end{restatable}
            \begin{proof} \hfill \\
                See Appendix~\ref{appendix:proofUpperBound} for the full proof.
            \end{proof}

            The condition depends on the profile only through the single quantity $\sum_k \left( 1 + 2^{-k} \right) \lambda_k$, and not on the way the edges are distributed among the arities. As we show next, the lower bound has the same feature.

        \subsection{Second-moment lower bound}
            For the lower bound, we adapt the second-moment argument of Achlioptas and Moore~\cite{AM2002}. As observed in later work (see, e.g.,~\cite{DFG15}), the calculation can be significantly simplified by focusing on \emph{equitable} colorings, whose two color classes differ in size by at most one.

            \begin{restatable}[Non-uniform lower bound for 2-colorability (second-moment method)]{lem}{nonUniformSecondMoment} 
                \label{lem:nonUniformSecondMoment}
                \hfill \\
                For every $\varepsilon > 0$, there exists $k_0 > 0$ such that the following holds.  Let $\HModel$ be a hypergraph distribution such that $\mathcal{M}(k) \leq \lambda_k \cdot 2^k \cdot n$ for every $k$. If
                \[
                    \sum_k \frac{\lambda_k}{1 - 2^{-k}} < (1 - \varepsilon)\frac{\ln 2}{2}
                    \text{\quad and \quad}
                    \min \{ k: \lambda_k > 0 \} \geq k_0
                \]
                then, for a hypergraph $\mathcal{H} \sim \HModel$
                \[
                    \liminf_{n \to \infty} \Prob{\mathcal{H} \text{ is 2-colorable}} > 0
                \]
            \end{restatable}
            \begin{proof} \hfill \\
                See Appendix~\ref{appendix:proofLowerBound} for a detailed proof.
            \end{proof}

        \subsection{Sharp-threshold upgrade}

            To strengthen Lemma~\ref{lem:nonUniformSecondMoment} to an a.a.s. statement, we show that $2$-colorability attains a sharp threshold under our non-uniform model for each fixed density profile. We embed the corresponding independent-edge model into the random CSP model of Hatami and Molloy~\cite{HM08}, which enables us to apply their necessary condition for a coarse threshold. We then follow Friedgut's argument for the uniform model~\cite{Friedgut2005} to rule out a coarse threshold.

            \phantomsection
            \begin{restatable}[Non-uniform lower bound for 2-colorability]{thm}{nonUniformLowerBound} 
                \label{thm:nonUniformLowerBound}
                \hfill \\
                For every $\varepsilon > 0$, there exists $k_0 > 0$ such that the following holds.  Let $\HModel$ be a hypergraph distribution such that $\mathcal{M}(k) \leq \lambda_k \cdot 2^k \cdot n$ for every $k$. If
                \[
                    \sum_k \frac{\lambda_k}{1 - 2^{-k}} < (1 - \varepsilon)\frac{\ln 2}{2}
                    \text{\quad and \quad}
                    \min \{ k: \lambda_k > 0 \} \geq k_0
                \]
                then, a hypergraph drawn from $\HModel$ is $2$-colorable a.a.s.
            \end{restatable}
            See Appendix~\ref{appendix:sharpThreshold} for a detailed proof.

            When $\HModel$ reduces to $H_k(n, M)$, the resulting lower bound recovers the leading-order asymptotics of the uniform bound. The fixed relative margin in our statement does not recover the sharper lower-order terms in~\cite{AM2002,DFG15}, or those obtained by Coja-Oghlan and Zdeborová~\cite{CZ2012} via an analysis of the \emph{condensation transition}.

    \section{Algorithmic 2-coloring} 

        We now analyze AKKT on non-uniform instances drawn from $\HModelExpected$. The algorithm operates on the current coloring of an edge rather than on its original size, so it extends to this setting without any modification. Our analysis follows the same general line as the original one, while simplifying and generalizing several of its arguments. Non-uniformity nevertheless leads to an interesting behaviour unseen in the existential case.

        \subsection{The algorithm and its success condition}
        For convenience, assume that $n$ is even. The algorithm performs $n/2$ rounds, in each of which it selects two uncolored vertices, colors one \emph{RED}, and the other \emph{BLUE}. Call a $k$-edge an $l$-tail if exactly $l$ of its vertices remain uncolored and all $k-l$ already-colored vertices have the same color. The algorithm declares failure whenever a $1$-tail or a monochromatic edge appears. To avoid it, the algorithm focuses on $2$-tails and $3$-tails.

        \algblockdefx{TB}{ETB}[1]{\textbf{do} #1 times }{\textbf{end do}}
        \begin{algorithm}[H]
            \begin{algorithmic}[1]
                \TB{$n/2$}
                    \If{there exists a $3$-tail or $2$-tail $e$}
                        \State pick any two uncolored vertices $x, y \in e$ and color $x$ \emph{RED} and $y$ \emph{BLUE}.
                    \Else
                        \State pick any two uncolored vertices $x,y$ and color $x$ \emph{RED} and $y$ \emph{BLUE}.
                    \EndIf
                \ETB
            \end{algorithmic}
            \caption{AKKT algorithm}
        \end{algorithm}

        The whole analysis rests on a single quantity, namely the expected number of $3$-tails created in one round. As long as it remains below $1$ by a constant margin, the algorithm succeeds. The structural properties we establish along the way, such as short positive intervals, few concurrently active edges and the lack of complex intersections between them, are consequences of this bound rather than separate assumptions.

        \begin{restatable}[AKKT success condition]{thm}{bounded3Tails} 
            \label{thm:bounded3Tails}
            \hfill \\
            Let $\HModelExpected$ be a hypergraph distribution such that $\min \{ k: \mathcal{P}(k) > 0 \} \geq 5$ and the expected number of newly created $3$-tails after each round of AKKT is bounded by $1 - \varepsilon$ for some $\varepsilon > 0$. Then AKKT succeeds a.a.s.
        \end{restatable}

        \subsection{Dynamic exposure and new 3-tails}
            To analyze AKKT on non-uniform instances, view its execution as a dynamic exposure process. Since edges are sampled independently in $\HModelExpected$, instead of starting with the whole graph, we can reveal them as they become 3-tails. 
            
            Throughout the execution, we will track an \emph{auxiliary} $3$-uniform multihypergraph on the same set of vertices as the one AKKT tries to $2$-color. Initially, the auxiliary hypergraph is empty. Whenever an edge in the original hypergraph becomes a $3$-tail, we add a $3$-edge on its uncolored vertices to the auxiliary hypergraph.          

            Examine more closely $3$-tails created after the $j$-th step, that is, when exactly $2j+2$ vertices are colored. First, we narrow down the analysis to $3$-tails originating from edges of arity $k$. With $n-2j-2$ uncolored vertices left, there are $\binom{n-2j-2}{3}$ triples, which are at risk of becoming new $3$-tails. For such a triple to become a $3$-tail, it must be a part of some original $k$-edge, in which exactly one vertex was colored in the current round. Therefore, each triple can be completed to a potential $k$-edge in $2\binom{j}{k-4}$ ways. Note that by this procedure, each $k$-subset of the vertices will be examined at most once. This enables us to delay the \emph{coin flip} for each edge until some triple is at risk of becoming a $3$-tail through it. The following lemma greatly simplifies the subsequent analysis.

            \begin{restatable}[Expected number of new $3$-tails]{lem}{algorithmicPoisson} 
                \label{lem:algorithmicPoisson} 
                \hfill\\ 
                Suppose that a hypergraph is sampled from $H_k(n,p_k)$, where \[ \binom{n}{k}p_k = \lambda_k \frac{2^k}{k}n, \] with $k\geq 5$. Let $j=j(n)$ and set \[ \alpha_n:=\frac{2j(n)}{n}. \] Let $X$ count the $3$-tails appearing after the $j$-th step of AKKT. Then 
                \[ 
                    \EV{X} = \frac{16}{3}\lambda_k(k-1)(k-2)(k-3) \alpha_n^{k-4}(1-\alpha_n)^3+o(1) 
                \] 
                as $n\to\infty$. 
            \end{restatable}

            \begin{proof}
                \hfill \\
                Since $X \sim \operatorname{Bin}\left( 2\binom{n-2j-2}{3}\binom{j}{k-4}, p_k\right)$,
                \[
                    \EV{X} = \binom{n-2j-2}{3} \cdot 2\binom{j}{k-4} \cdot p_k
                    \sim \frac{(n-2j-2)^3}{6} \cdot 2 \frac{j^{k-4}}{(k-4)!} \cdot \lambda_k \frac{2^k}{k} \frac{n}{\binom{n}{k}}
                \]
                \[
                    \sim \frac{1}{3}n^3 \left( 1 - \frac{2j}{n} \right)^3 \frac{\left( \frac{2j}{n} \right)^{k-4}}{ (k-4)! } \left( \frac{n}{2} \right)^{k-4} \cdot \lambda_k \frac{2^k}{k} n \cdot \frac{k!}{n^{k}}
                    = \frac{16}{3} \lambda_k \left( 1 - \frac{2j}{n} \right)^{3} \left( \frac{2j}{n} \right)^{k-4} (k-1)(k-2)(k-3)
                \]
                \[
                    = \frac{16}{3} \lambda_k (1-\alpha_n)^3 \alpha_n^{k-4} (k-1)(k-2)(k-3)
                \]
                Note that for $\alpha_n = \oh(1)$ or $\alpha_n = 1 - \oh(1)$, the expected count of new 3-tails is $\oh(1)$. Treating the formula for $\EV{X}$  as a function of $\alpha_n$, differentiation shows that it is unimodal and maximized at $1 - \frac{3}{k-1}$.
            \end{proof}

        \subsection{Active 3-tails}
            An edge in the auxiliary hypergraph is called \emph{active} if it can still become monochromatic under the current partial coloring. We first bound the number of active edges and then control their intersections.
            
            A \emph{positive interval} is a maximal consecutive sequence of rounds during which at least one edge is active. The next proposition bounds both its length and the number of concurrently active edges.

            \begin{restatable}[Short positive intervals and few active edges]{prop}{shortIntervalsAlgorithmic} 
                \label{lem:shortIntervalsAlgorithmic}
                \hfill

                W.h.p., the length of any positive interval is at most $\ln^2 n$. Moreover, at most $\ln^2 n$ tails are created during each positive interval.
            \end{restatable}
            \begin{proof}
                \hfill 

                Recall that the count of new $3$-tails after each round is a sum of independent Bernoulli random variables, as is the total count of new $3$-tails in any interval.

                Suppose that some positive interval has length at least $\ln^2 n$. As in every round at least one active edge is chosen by the algorithm and ceases to be active, at least $\ln^2 n$ new $3$-tails must appear in the first $\ln^2 n$ steps of this interval, as otherwise, the interval would be shorter. The probability of such an event can be bounded using a Chernoff bound. Let $X$ be a random variable counting new $3$-tails in a specific interval of length $\ln^2 n$. Using only the fact that the expected number of new $3$-tails in each step is at most $1-\varepsilon$ we can bound the probability of $X$ being large as follows:
                \[
                    \Prob{ X \geq \ln^2 n }
                    \leq \left( \frac{e^{ \frac{\varepsilon}{1 - \varepsilon} }}{ \left( 1 + \frac{\varepsilon}{1 - \varepsilon} \right)^{1 + \frac{\varepsilon}{1 - \varepsilon}} } \right)^{(1-\varepsilon)\ln^2 n}
                    = \left( (1-\varepsilon)e^\varepsilon \right)^{\ln^2 n}
                    \leq \exp \left(-\frac{\varepsilon^2}{2} \ln^2 n \right) = \oh\left( n^{-1} \right)
                \]
                where the last inequality uses $\ln \left(1 - \varepsilon\right) \leq -\varepsilon - \frac{\varepsilon^2}{2}$.

                There are at most $n/2$ starting rounds for such an interval, therefore by the union bound, the probability that any interval is of length at least $\ln^2 n$ is $\oh(1)$. By the same reasoning we can bound the total count of active edges in each interval.
            \end{proof}

            Since only a few edges are active at the same time, complex intersections between them are unlikely as long as sufficiently many vertices remain uncolored. Next, we show that with high probability, no tails appear in the final rounds. Consequently, whenever we see a new tail, we may safely assume that there are many uncolored vertices left.

            \begin{restatable}[Final rounds are tail-free]{prop}{tailFreeAlgorithmic} 
                \label{lem:tailFreeAlgorithmic}
                \hfill

                W.h.p., there are no new active edges in the final $\frac{1}{2}n^{2/3}$ rounds.
            \end{restatable}
            \begin{proof}
                \hfill

                For every edge arity $k$, the expected number of new $3$-tails as a function of $\alpha$ is unimodal and attains its maximum long before the interval under consideration.
                Thus, by Lemma~\ref{lem:algorithmicPoisson}, the expected count of new $3$-tails in the final $\frac{1}{2}n^{2/3}$ rounds can be bounded by 
                \[
                    \frac{1}{2}n^{2/3} \cdot \sum_k \frac{16}{3} \lambda_k \left( n^{-1/3} \right)^3 \left( 1 - n^{-1/3} \right)^{k-4} (k-1)(k-2)(k-3)
                    = \Oh \left( n^{-1/3} \right)
                \]
            \end{proof}

            A $1$-tail can appear for two reasons: two active tails may share two vertices, or one tail may intersect two others on two different vertices. The following proposition utilizes the abundance of uncolored vertices to show that both configurations are unlikely.

            \begin{restatable}[Active tails have no complex intersections]{prop}{simpleIntersectionsAlgorithmic} 
                \label{lem:simpleIntersectionsAlgorithmic}
                \hfill 

                W.h.p., throughout the execution, every active tail intersects at most one other active tail, and any two tails have at most one common vertex.
            \end{restatable}
            \begin{proof}
                \hfill 

                Recall that w.h.p. at most $\ln^2 n$ tails appear in each positive interval, and every new tail appears while at least $\frac{1}{2} n^{2/3}$ vertices remain uncolored.

                A given pair of tails shares at least two vertices with probability $\Oh(n^{-4/3})$. 
                
                Now consider three tails, and their order of appearance. If the first two are disjoint, the third must meet both, with probability $\Oh(n^{-4/3})$. Otherwise, the first intersection and the third tail meeting their union each contribute a factor $\Oh(n^{-2/3})$. Therefore, the probability that one of the three tails intersects both others can again be bounded by $\Oh(n^{-4/3})$.

                Taking a union bound over all pairs and triples in each positive interval, and over at most $n/2$ intervals, completes the proof.
            \end{proof}

            Since an edge must become a $1$-tail before becoming monochromatic, this proves Theorem~\ref{thm:bounded3Tails}. 

        \subsection{An explicit sufficient condition}
            With the help of Lemma~\ref{lem:algorithmicPoisson}, we can express the expectation in terms of the density profile. The following is a direct consequence of the expectation bound in Theorem~\ref{thm:bounded3Tails}.

            \begin{restatable}[Explicit AKKT success condition]{thm}{AKKTsuccess} 
                \label{thm:AKKTsuccess}
                \hfill 

                Let $\HModelExpected$ be a hypergraph distribution such that $\binom{n}{k}\mathcal P(k)=\lambda_k\frac{2^k}{k}n$ for all $k$ and $\min\{k:\lambda_k>0\}\geq5$. If, for some $\varepsilon>0$,
                \[
                    \max_{x\in[0,1]}\left(\frac{16}{3}(1-x)^3
                    \sum_{k\geq5}\lambda_k(k-1)(k-2)(k-3)x^{k-4}\right)\leq1-\varepsilon,
                \]
                then AKKT succeeds a.a.s. on a hypergraph sampled from $\HModelExpected$.
            \end{restatable}

            \begin{proof}
                \hfill
            
    Here $x$ is the fraction of vertices already colored, and the expression being maximized is the asymptotic expected number of new $3$-tails per round. Equivalently, defining
            \[
                f(x):=\sum_k\frac{\lambda_k}{k}x^k,
            \]
            the condition is $\max_{x\in[0,1]}\frac{16}{3}(1-x)^3f^{(4)}(x)\leq1-\varepsilon$.

                Setting $x = \frac{2j}{n}$ and explicitly computing $f^{(4)}(x)$ yields the same expected count as that derived in Lemma~\ref{lem:algorithmicPoisson}.
            \end{proof}

    \section{Benefits of non-uniformity}
        The condition of Theorem~\ref{thm:AKKTsuccess} is not a bound on a single number. In this section we show that this is not merely an artifact of our analysis, as no condition expressed in terms of the normalized density alone can capture the behaviour of AKKT.

        \subsection{The uniform baseline}
            With Theorem~\ref{thm:AKKTsuccess} in hand, we may compute an explicit bound for the $k$-uniform $\HUniformExpected$ model. With the algorithmic weighting, write the normalized density as $\lambda$, so that $p \binom{n}{k} = \lambda \cdot \frac{2^k}{k} n$. Optimizing, we obtain that as $k \to \infty$, the largest $\lambda$ allowed by this condition approaches $\frac{e^3}{144} \approx 0.1395$ with $\lambda \leq 0.13$ sufficing for all $k \geq 24$. Additionally, $\lambda \leq 0.074$ suffices for all $k \geq 5$.

            The authors of \cite{AKKT02} also show that for sufficiently large $\lambda$, the algorithm fails a.a.s.\footnote{In fact, their argument shows $\lambda \geq 1$ suffices.} In the uniform model, the normalized density therefore determines the outcome, up to the constant separating these two bounds.

        \subsection{Arbitrarily large normalized density}
            The next theorem shows that sufficiently separated edge sizes allow AKKT to succeed a.a.s. at arbitrarily large normalized density.

            \begin{restatable}[Normalized density can be arbitrarily large]{thm}{UnboundedWeight}
                \label{thm:UnboundedWeight} 
                \hfill \\
                For every $c > 0$, there exist a density profile $\left( \lambda_k \right)_{k \geq 1}$ and a distribution $\HModelExpected$ satisfying $\mathcal{P}(k) \binom{n}{k} =  \lambda_k \cdot \frac{2^k}{k} n$ and
                \[
                    \sum_k \lambda_k > c
                \]
                for which AKKT succeeds a.a.s. on a hypergraph drawn from $\HModelExpected$.
            \end{restatable}

            Mixing edge sizes may therefore bring an unbounded algorithmic advantage. At a normalized density for which AKKT is known to fail on uniform instances, it still succeeds a.a.s. on the instances constructed below.

            Figure~\ref{fig:singleArities} plots the expectation from Lemma~\ref{lem:algorithmicPoisson} for three edge arities. It illustrates how the contribution from a single arity varies during the execution.

            \begin{figure}[H]
                \centering
                \begin{subfigure}[b]{0.3\textwidth}
                    \centering
                    \includegraphics[width=\textwidth]{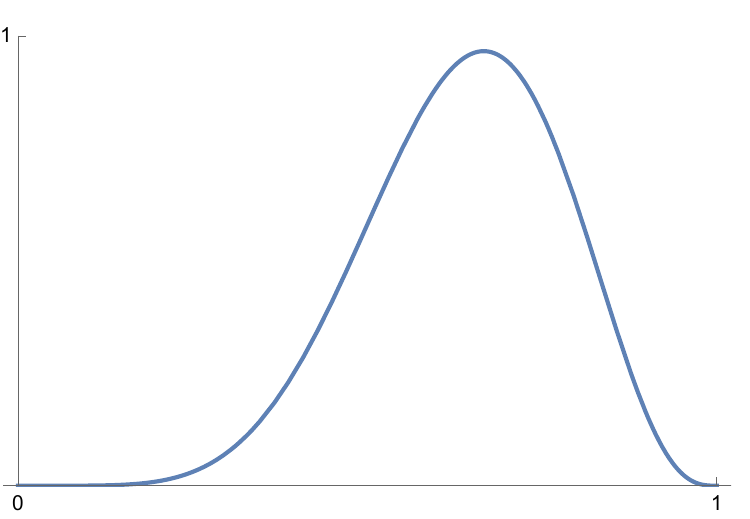}
                    \caption{$k = 10$}
                \end{subfigure}
                \hfill
                \begin{subfigure}[b]{0.3\textwidth}
                    \centering
                    \includegraphics[width=\textwidth]{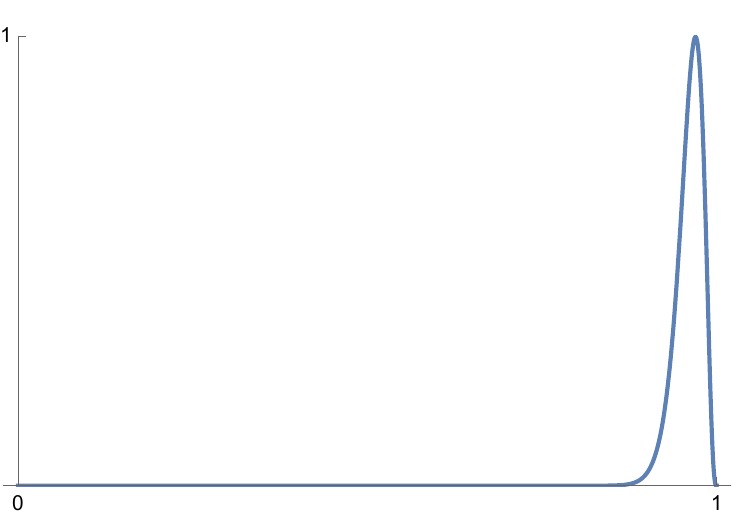}
                    \caption{$k = 100$}
                \end{subfigure}
                \hfill
                \begin{subfigure}[b]{0.3\textwidth}
                    \centering
                    \includegraphics[width=\textwidth]{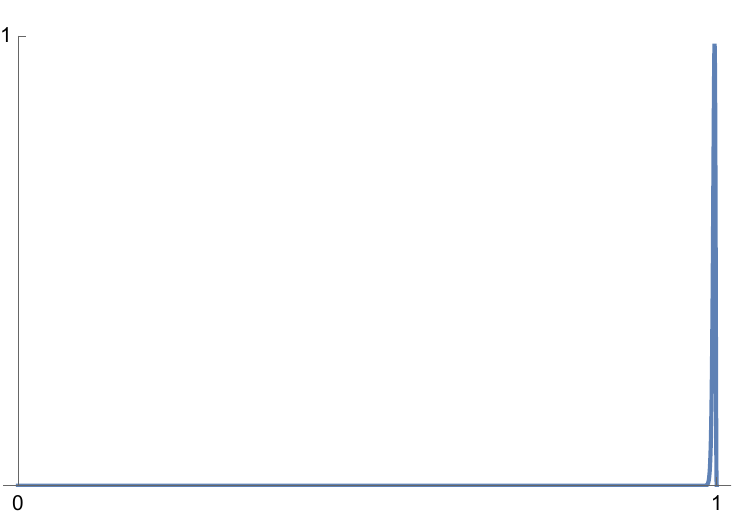}
                    \caption{$k = 1000$}
                \end{subfigure}
                \caption{Expected count of new $3$-tails throughout AKKT's execution.}
                \Description{Three plots show the expected number of new 3-tails for edge arities 10, 100, and 1000. As the arity increases, the peak moves later in the execution and becomes narrower.}
                \label{fig:singleArities}
            \end{figure}

            Observe that as $k$ increases, the peak occurs later and the interval with non-negligible expectation gets shorter. Because of this, $3$-tails originating from edges which are far apart in size are not likely to appear in the same positive interval. We build on this in the proof of Theorem~\ref{thm:UnboundedWeight} by picking sufficiently spaced edge arities. Figure~\ref{fig:multiArities} plots the total expectation of new $3$-tails for several non-uniform distributions. In each of these examples, the contribution coming from each arity peaks just below $1$.
            \begin{figure}[H]
                \centering
                \begin{subfigure}[b]{0.3\textwidth}
                    \centering
                    \includegraphics[width=\textwidth]{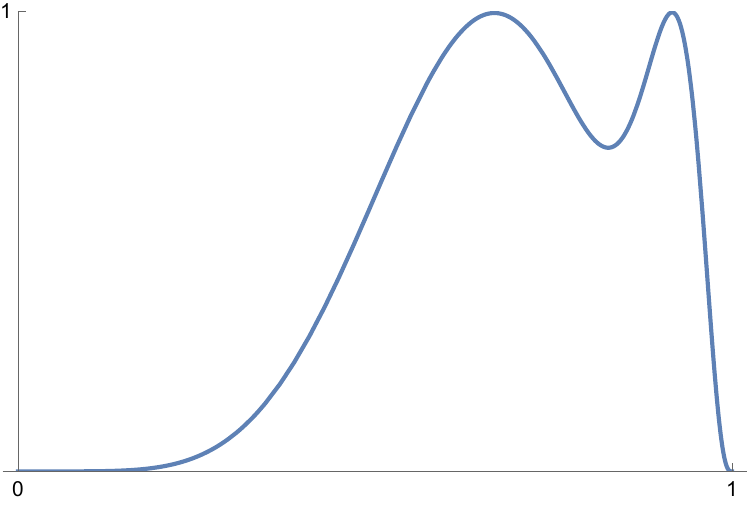}
                    \caption{$k_0 = 10, k_1 = 40$}
                \end{subfigure}
                \hfill
                \begin{subfigure}[b]{0.3\textwidth}
                    \centering
                    \includegraphics[width=\textwidth]{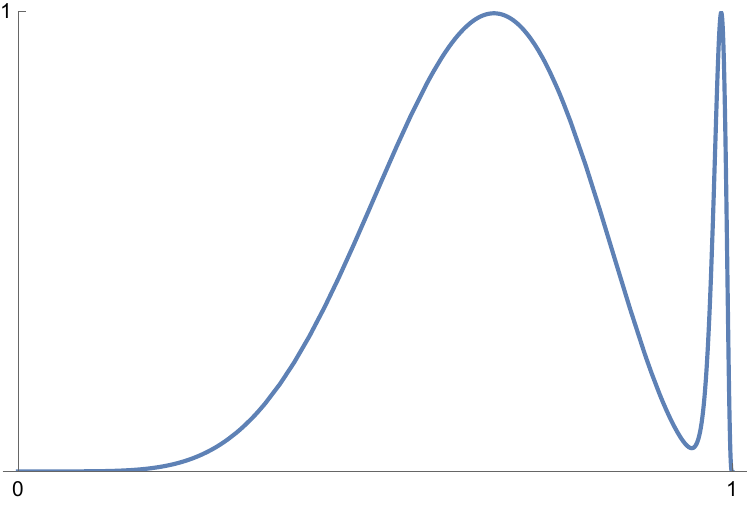}
                    \caption{$k_0 = 10, k_1 = 200$}
                \end{subfigure}
                \hfill
                \begin{subfigure}[b]{0.3\textwidth}
                    \centering
                    \includegraphics[width=\textwidth]{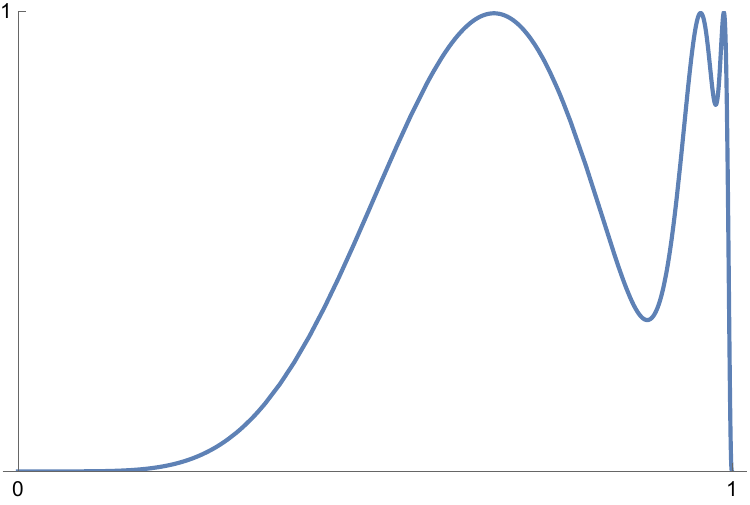}
                    \caption{$k_0 = 10, k_1 = 70, k_2 = 300$}
                \end{subfigure}
                \caption{Expected count of new $3$-tails throughout AKKT's execution for non-uniform distributions.}
                \Description{Three plots show the combined expected number of new 3-tails for profiles supported on arities 10 and 40; 10 and 200; and 10, 70, and 300. They illustrate the effect of separating edge sizes on the overlap between their contributions.}
                \label{fig:multiArities}
            \end{figure}

            We now make this separation precise in the proof of Theorem~\ref{thm:UnboundedWeight}.
            \begin{proof}
                \hfill \\
                Recall that the expected number of new $3$-tails originating from $k$-edges is unimodal and maximized at $\alpha = 1 - \frac{3}{k-1}$, where $\alpha = \frac{2j}{n}$ parametrizes the round number. Let $\left( 1 - \frac{1}{\sqrt{k}}, 1 - \frac{1}{k^2} \right)$ be the critical interval for edges of arity $k$. We show that the contribution from $k$-edges is negligible whenever $\alpha$ lies outside this interval.

                Substituting $\alpha = 1 - \frac{1}{\sqrt{k}}$ into the expression from Lemma~\ref{lem:algorithmicPoisson}, we obtain
                \[
                    \frac{16}{3}\lambda_k (k-1)(k-2)(k-3) \left( \frac{1}{\sqrt{k}} \right)^3 \left( 1 - \frac{1}{\sqrt{k}} \right)^{k-4}
                    < C_1 \cdot k^{3/2} \cdot e^{-\sqrt{k}}
                \]

                Similarly, substituting $\alpha = 1 - \frac{1}{k^2}$ gives
                \[
                    \frac{16}{3}\lambda_k (k-1)(k-2)(k-3) \left( \frac{1}{k^2} \right)^3 \left( 1 - \frac{1}{k^2} \right)^{k-4}
                    < C_2 \cdot k^{-3}
                \]

                Therefore, for sufficiently large $k$, the contribution from $k$-edges outside its critical interval is bounded by $C k^{-3}$, for some $C > 0$. The key to success is to choose a sequence of arities such that their critical intervals are disjoint, for example $k_{i+1} = k_i^4$. Then, every $\alpha$ lies in at most one critical interval, while the combined contribution from all remaining selected arities is bounded by $C \sum_{k \geq k_0} k^{-3}$. Since $\sum_k k^{-3}$ converges, this can be made arbitrarily small by choosing sufficiently large $k_0$.

                Finally, fix a small $\varepsilon > 0$ and assign $\lambda_{k_i} = \lambda > 0$ such that the peak contribution from each selected arity $k_i$ is bounded by $1-2\varepsilon$. By taking $k_0$ sufficiently large, $C \sum_{k \geq k_0} k^{-3} \leq \varepsilon$. Since the total expected count of new $3$-tails is below $1 - \varepsilon$, AKKT succeeds a.a.s. As we can pick arbitrarily many different arities, the proof is complete.
            \end{proof}

    \section{Consequences of non-uniformity and open questions}
        Perhaps the most interesting consequence of non-uniformity is the contrast between the algorithmic and the non-constructive results. In the non-constructive setting, both bounds essentially depend on a purely quantitative condition, namely the expected number of monochromatic edges. By contrast, the condition for AKKT's success is subtler and more qualitative in spirit, as what matters is not how much constraint the instance carries, but how it is distributed over the execution. This raises the question of whether this contrast reflects a fundamental difference between the two settings, or is specific to AKKT.

        There are reasons to believe that the algorithmic lower bound cannot be improved to $\omega \left( \frac{2^k}{k}\ln k \cdot n \right)$. Evidence for such a barrier comes from the geometry of the \emph{solution space}. Achlioptas and Coja-Oghlan \cite{ACO08} showed that for random $k$-graphs with an expected number of edges between $(1 + \varepsilon_k) \frac{2^{k-1}}{k} \ln k \cdot n$ and $(1 - \varepsilon_k) 2^{k-1} \ln 2 \cdot n$, this space is \emph{shattered} a.a.s., where $\varepsilon_k\to0$ as $k\to\infty$. That is, the proper $2$-colorings are scattered into many small regions that are separated by large Hamming distances. Similar phenomena occur in a number of related problems, often coinciding with the point at which the best-known polynomial-time algorithms fail. Moreover, similar behaviour has been proved to be a barrier for certain classes of algorithms, see~\cite{Gamarnik21}.

        For the closely related random $k$-SAT, Coja-Oghlan \cite{CO10} provided an algorithm that succeeds up to the leading-order asymptotic scale of the shattering threshold. If an analogous algorithm were found for $2$-colorability, it would be natural to ask whether it would extend to the non-uniform setting. It would be particularly interesting to determine whether its success condition depends only on the normalized density under a suitable weighting or, as for AKKT, on the full density profile.

        Recall that Friedgut's result \cite{Friedgut1999} implies that, for each $k \geq 3$, there exists a sharp threshold sequence $r_{col}(k, n)$. In particular, for every $\varepsilon > 0$, hypergraphs sampled from $\HUniform$ are $2$-colorable a.a.s. if $M < (1-\varepsilon)r_{col}(k, n) n$ and non-$2$-colorable a.a.s. if $M > (1+\varepsilon)r_{col}(k, n)n$. It is an open problem to determine whether a \emph{universal} $r_{col}(k) := \lim_{n \to \infty} r_{col}(k, n)$ exists. Coja-Oghlan and Zdeborová \cite{CZ2012} conjecture that $r_{col}(k) = \frac{\ln 2}{2} \left(2^k - 1\right)-\frac{1}{4}+\oh(1)$, as $k \to \infty$. In the non-uniform setting, an interesting and perhaps even harder question is whether the threshold can be characterized by a single quantity, perhaps obtained by appropriate weighting.

    \section*{AI statement}
        The motivation for this paper, the proof strategies, as well as the majority of the computations come directly from the authors. The only part whose development was partially AI-assisted is Appendix~\ref{appendix:sharpThreshold}, on which we elaborate below.

        Generative AI was used in the finalization of the paper as an editing aid, to help us proofread, check proofs, and find and correct mistakes. In the process, we used some AI-generated suggestions for minor simplifications of certain computations, without changing their ``original spirit'' or altering the proof strategies.

        Regarding the aforementioned Appendix~\ref{appendix:sharpThreshold}, part of the story is already described in the paper. When we started work on this project, roughly two years ago, the first proof strategy for the sharp-threshold transition that we tried was almost identical to the one presented in the paper. However, as we explain there, the tools we use are very delicate in this context. Unfortunately, we managed to convince ourselves that this approach did not work. For this reason, we took a much more technical and longer approach. However, during the final, AI-assisted review of the paper, AI suggested a ``similar'' approach. The approach suggested by AI was misleading, but it prompted us to once again examine, with the help of AI, our original approach. To our surprise, the original idea worked, and we were able to drastically simplify the proof. Thus, the proof presented in the appendix is difficult to classify simply as AI-generated or not: the initial idea was ours, but it was revived by an indirect AI suggestion. Due to this ambiguity, we opted to describe the process in greater detail.

    \patchcmd{\thebibliography}{\section*{\refname}}{}{}{}
    \bibliographystyle{unsrt}
    \bibliography{pb}

    \appendix
    
    \section{Sharp thresholds}
    \label{appendix:sharpThreshold}
        Recall that $\HModelExpected$ denotes the independent-edge model determined by a function $\Ps:\Nat \rightarrow [0,1]$. For $c>0$, write $(c\Ps_n)(k):=c\Ps_n(k)$.

        A hypergraph property $\mathcal{A} = (\mathcal{A}_n)_{n \geq 1}$ is \emph{monotone} if it is preserved under edge addition, and \emph{symmetric} if it is invariant under vertex permutations.

        We call $(\Ps_n)$ a \emph{sharp threshold profile} for $\mathcal{A}$ if for every $0 < \varepsilon < 1$
        \[
            \Prob{ \widehat H(n;(1-\varepsilon)\Ps_n) \in \mathcal{A}_n } \to 0
        \]
        and
        \[
            \Prob{ \widehat H(n;(1+\varepsilon)\Ps_n) \in \mathcal{A}_n } \to 1
        \]

        \subsection{Embedding $\HModelExpected$ into Hatami-Molloy's model}
        Hatami and Molloy~\cite{HM08} extended Friedgut's sharp threshold tools~\cite{Friedgut1999, Friedgut2005} to the following random CSP model.
        \begin{definition*}{$\RandomCSP$ model} \hfill \\
            Let $\mathcal{Q}$ be a distribution on $K$-ary constraints. For each ordered $K$-tuple of distinct variables chosen from $n$ variables and each constraint $C \in \operatorname{supp}(\mathcal{Q})$, impose $C$ on the tuple independently with probability $\mathcal{Q}(C)\frac{p}{K!}$.
        \end{definition*}
        
        Fix a nonzero density profile $(\lambda_k)_{k \geq 3}$ and set $\Ps_n(k) = \lambda_k 2^k n / \binom{n}{k}$. We first show how to embed the corresponding $\HModelExpected$ model into $\RandomCSP$.
        
        Let $K := \max\{ k: \lambda_k > 0 \}$.For each $k$ with $\lambda_k>0$, let $C_k$ be the boolean $K$-ary constraint requiring that the first $k$ variables are not all equal, with no restriction on the remaining variables. This is precisely the restriction imposed by a $k$-edge. Define the distribution $\mathcal{Q}$ by
        \[
            \mathcal{Q}(C_k) = \frac{\lambda_k 2^k}{\sum_j \lambda_j 2^j} 
        \]
        \begin{restatable}[Embedding $\HModelExpected$ into $\RandomCSP$]{prop}{modelEmbedding} 
            \label{prop:modelEmbedding}
            \hfill

            For $\mathcal Q$ defined above and
            \[
                p := \frac{n \sum_j\lambda_j 2^j}{\binom{n}{K}},
            \]
            the hypergraph obtained from $\RandomCSP$ has independent edges, and each potential $k$-edge of a permitted size appears with probability $(1+\oh(1))\Ps_n(k)$ as $n \to \infty$.
        \end{restatable}
        \begin{proof}
            \hfill 

            Fix a potential $k$-edge $e$ with $\lambda_k > 0$. This edge appears if at least one occurrence of $C_k$ has exactly the vertices of $e$ in its first $k$ positions. There are $k!$ ways to order these vertices and $(n-k)!/(n-K)!$ ways to fill the remaining positions with distinct vertices outside $e$. Thus $\frac{k!(n-k)!}{(n-K)!}$ possible constraint occurrences produce $e$. Each is included independently with probability $\mathcal{Q}(C_k)\frac{p}{K!} = \frac{\lambda_k 2^k n}{K!\binom{n}{K}}$. Therefore
            \[
                \Prob{e \text{ appears}}
                = 1-\left(1-\frac{\lambda_k2^k n}{K!\binom nK}\right)^{\frac{k!(n-k)!}{(n-K)!}}
                \geq 1 - \exp(-\Ps_n(k))
                \geq \Ps_n(k)-\frac{\Ps_n(k)^2}{2}
            \]
            We used $1-x\leq\exp(-x)$ and $\exp(-x)\leq1-x+x^2/2$. The union bound gives $\Prob{e \text{ appears}} \leq \Ps_n(k)$, so
            \[
                0 \leq \Ps_n(k) - \Prob{e \text{ appears}}
                \leq \frac{\Ps_n(k)^2}{2}.
            \]
            Since $\Ps_n(k) \to 0$, this proves the claimed asymptotic.

            To compare the two random hypergraphs, assign each potential edge an independent uniform random number from $[0,1]$. Use the same number in both models, including the edge whenever it is below the respective inclusion probability. A $k$-edge then appears in exactly one model with probability equal to the difference between the two inclusion probabilities, which is at most $\Ps_n(k)^2/2$.

            By a union bound over all potential edges, the probability that the two hypergraphs differ is at most
            \[
                \frac{1}{2} \sum_k \binom{n}{k} \Ps_n(k)^2
                = \frac{n^2}{2} \sum_k \frac{(\lambda_k 2^k)^2}{\binom{n}{k}}
                = \Oh(n^{-1})
                = \oh(1)
            \]
            Consequently, any hypergraph property, including $2$-colorability, holds in the two models with probabilities differing by $\oh(1)$.
        \end{proof}

        The embedding gives the following consequence of Theorem~6 of~\cite{HM08}. If non-$2$-colorability\footnote{We use non-$2$-colorability because it is preserved under edge addition.} in $\HModelExpected$ has a coarse threshold, then there exist a sequence $(\Ps_n)$, constants $\tau,\alpha,\varepsilon > 0$, and a fixed hypergraph $M$ with permitted edge sizes such that, for infinitely many $n$:
        \begin{enumerate}
            \item $\alpha < \Prob{ \HModelExpected \text{ is non-2-colorable} } < 1 - 3\alpha$ 
            \item $\Prob{ \widehat{H}(n; (1 + \varepsilon)\mathcal{P}) \text{ is non-2-colorable} } < 1 - 2\alpha$
            \item $\Prob{\HModelExpected \oplus M \text{ is non-2-colorable}} > 1 - \alpha$
            \item $\tau < \Prob{M \subseteq \HModelExpected}$
        \end{enumerate}
        Here $\oplus$ denotes adding a uniformly random copy of $M$ to the sampled hypergraph.

        We now follow Friedgut's proof~\cite{Friedgut2005} to rule out a coarse threshold. The argument extends to the non-uniform model without substantive changes, so we omit the details. 
        
        Recall that, conditional on the number of edges of each size, $\HModelExpected$ is exactly the corresponding fixed edge model. Moreover, these counts are concentrated around their mean, which is $\Theta(n)$ with $\oh(n)$ deviations a.a.s. By monotonicity, the colorability probability in $\HModel$ lies, up to $\oh(1)$, between those in independent-edge models with expected counts a fixed proportion above and below the prescribed counts. Since this proportion can be arbitrarily small, sharpness transfers to $\HModel$.

        Since the threshold is sharp, a suitable choice of $\varepsilon$ in Lemma~\ref{lem:nonUniformSecondMoment} gives $2$-colorability a.a.s. and proves Theorem~\ref{thm:nonUniformLowerBound}.

        \subsection{Note on sharp thresholds in non-uniform random structures}
        The embedding above enabled us to apply well-understood tools for uniform structures. Its construction is delicate and relies on all edge counts having the same asymptotic order in $n$. Otherwise, a distribution $\mathcal{Q}$ independent of $n$ could not be constructed.

        Finding the correct embedding delayed the completion of this paper. A similar construction was among our first proof strategies and appeared in an early draft, but it was incorrect. The difficulties with that construction led us to believe that this approach could not work. With that route apparently closed, we examined the proofs of Friedgut's tools~\cite{Friedgut1999,Friedgut2005} and observed that the same strategy extends to non-uniform hypergraphs. We worked out this extension intending to use it here. Fortunately, the embedding above ultimately made this unnecessary.

        Nevertheless, we present the results obtained by this approach, which apply to much broader families of non-uniform distributions, without requiring all edge counts to have the same asymptotic order. The proofs follow Friedgut's strategy and are available in a technical note on the author's website \cite{Ryn2026note}.
        
        Below we consider sequences $(\Ps_n)_{n \geq 1}$ with fixed finite support whose probabilities are \emph{proportional to powers of $n$}. That is, for each permitted size $k$,
        \[
            \Ps_n(k) = a_k n^{-\beta_k}
        \]
        for all sufficiently large $n$, where $a_k,\beta_k>0$ are constants independent of $n$.

        For a fixed hypergraph $F$ with permitted edge sizes, let $E_{\Ps_n}(F)$ denote the expected number of copies of $F$ in $\widehat H(n;\Ps_n)$, where copies need not be induced. We call $F$ \emph{balanced} with respect to $(\Ps_n)$ if, for every subhypergraph $F' \subseteq F$ containing at least one edge,
        \[
            E_{\Ps_n}(F) = \Oh\left(E_{\Ps_n}(F')\right) \qquad \text{as } n \to \infty.
        \]

        \begin{restatable}[Non-uniform coarse threshold theorem]{thm}{applicabletwo}
            \label{thm:applicable2}
            \hfill

            Let $\mathcal A=(\mathcal A_n)_{n\geq1}$ be a monotone symmetric hypergraph property, and let $(\Ps_n)$ be a sequence as above. Suppose that there exist constants $0<\eta<1$, $\rho>0$, and an infinite set $\mathcal{N} \subseteq \Nat$ such that, for every $n \in \mathcal{N}$,
            \[
                \rho
                < \Prob{\widehat H(n;(1-\eta)\Ps_n)\in\mathcal A_n}
                \leq \Prob{\widehat H(n;(1+\eta)\Ps_n)\in\mathcal A_n}
                < 1-\rho.
            \]

            Then there exist constants $\varepsilon,\alpha,\tau>0$, an infinite set $\mathcal N'\subseteq\mathcal{N}$, and a fixed hypergraph $M$ with permitted edge sizes balanced with respect to $(\Ps_n)$, such that
            \[
                E_{\Ps_n}(M)=\Theta(1).
            \]
            For every $n\in\mathcal N'$,
            \[
                \alpha
                < \Prob{\widehat H(n;\Ps_n)\in\mathcal A_n}
                < 1-3\alpha
            \]
            and
            \[
                \Prob{\widehat H(n;\Ps_n)
                    \text{ contains a copy of }M}
                \geq\tau.
            \]
            Moreover, for each such $n$, there exists a deterministic hypergraph $G$ on $[n]$ such that
            \[
                \Prob{G \oplus M \in \mathcal A_n} > 1-\alpha
            \]
            and
            \[
                \Prob{G \cup \widehat H(n;\varepsilon\Ps_n) \in\mathcal A_n} < 1-2\alpha.
            \]
        \end{restatable}

    \section{Proofs of non-constructive results}

    \subsection{\nameref{thm:nonUniformUpperBound}}
    \label{appendix:proofUpperBound}

    To avoid cumbersome floor and ceiling notation, we do all the calculations assuming that $4|n$. Introducing floors and ceilings changes only lower-order terms, so all asymptotic estimates remain unchanged.

    To further simplify the proofs, we treat $\HModel$ as if the edges were sampled with replacement. As the edge count is $\Oh(n)$ and every arity is at least $3$, the probability that any edge is selected more than once is $\oh(1)$, and hence it does not affect the asymptotic statements.

    \nonUniformUpperBound*
    \begin{proof} \hfill \\ 
        Let $\mathcal{H} \sim \HModel$ be a random hypergraph. Fix a $2$-coloring $\phi$ with $a$ red and $b := n - a$ blue vertices, and let $p_k$ denote the probability that a random $k$-edge is monochromatic under $\phi$. Then
        \[
            p_k = \frac{ \binom{a}{k} + \binom{b}{k} }{ \binom{n}{k} }
            \geq \frac{ 2 \binom{n/2}{k} }{ \binom{n}{k} }
            = 2 \frac{(n/2)^{\underline{k}}}{n^{\underline{k}}}= 2^{1-k}(1+\Oh(n^{-1}))  
        \] 
        where the last equality comes from the fact that the maximal arity $k$ with $\lambda_k>0$ is bounded by some constant. We can put an upper bound on the probability that $\phi$ is proper for $\mathcal{H}$.
        \[
            \Prob{ \phi \text{ is proper for } \mathcal{H} }
            \leq \prod_k (1 - p_k)^{\mathcal{M}(k)}
            \leq \prod_k \left( 1 - 2^{1-k} + \Oh(n^{-1}) \right)^{\mathcal{M}(k)}
            = \prod_k \left( (1 - 2^{1-k}) \left(1 + \Oh(n^{-1})\right) \right)^{\mathcal{M}(k)}
        \]
        \[
            = \exp\left( \sum_k \mathcal{M}(k) \ln \left( 1 - 2^{1-k} \right) \right) \cdot \exp\left( \sum_k \mathcal{M}(k) \ln \left( 1 + \Oh(n^{-1}) \right) \right)
        \]
        \[
            \leq \exp\left( \sum_k -\mathcal{M}(k) \left( 2^{1-k} + 2^{1-2k} \right) \right) \cdot \exp\left( \sum_k \mathcal{M}(k) \cdot \Oh(n^{-1}) \right)
        \]
        \[
            \leq \exp\left( \sum_k -2 \lambda_k \cdot n \cdot \left( 1 + 2^{-k} \right) \right) \cdot \Oh(1)
            = \exp\left( -2n \sum_k \left( 1 + 2^{-k} \right) \lambda_k \right) \cdot \Oh(1)
        \]
        Now, we can put an upper bound on the expected number of proper $2$-colorings of $\mathcal{H}$ 
        \[
            2^n \cdot \exp\left( -2n \sum_k \left( 1 + 2^{-k} \right) \lambda_k \right) \cdot \Oh(1) 
            = \exp\left( n \cdot \left( \ln 2 - 2 \sum_k \left( 1 + 2^{-k} \right) \lambda_k \right) \right) \cdot \Oh(1)
        \]
            Therefore, if $\sum_k \left( 1 + 2^{-k} \right) \lambda_k > \frac{\ln 2}{2}$ the expected number of proper $2$-colorings is $\oh(1)$. Hence, by Markov's inequality, a hypergraph sampled from $\HModel$ is non-$2$-colorable a.a.s.
        \end{proof}

    \subsection{\nameref{thm:nonUniformLowerBound}}
    \label{appendix:proofLowerBound}
    \nonUniformSecondMoment*
    \begin{proof} 
        \hfill \\
        Let $X$ be a random variable  counting the number of proper, equitable colorings of a hypergraph $\mathcal{H} \sim \HModel$. To show that the probability that $\mathcal{H}$ is $2$-colorable is positive, we will employ the second moment method. Recall that for nonnegative random variable with finite variance, Paley-Zygmund inequality implies 
        \[
            \Prob{ X > 0 } \geq \frac{\EV{ X }^2}{\EV{ X^2 }}.
        \]
        Therefore, it is sufficient to show $\frac{\EV{ X^2 }}{\EV{ X }^2} = \Oh(1)$. First, we calculate both $\EV{X}$ and $\EV{X^2}$.

        Let $M_k := \lambda_k \cdot 2^k \cdot n$. To simplify the subsequent computation, from now on we pessimistically assume that $\mathcal{H}$ has exactly $M_k$ $k$-edges.
        \begin{equation}
        \label{eq:EX}
            \EV{X} = \binom{n}{n/2} \prod_k \left( 1 - \frac{2 \binom{n/2}{k}}{ \binom{n}{k} } \right)^{M_k}
        \end{equation}
        \[
            \EV{X^2} = \sum_{ \substack{c_1, c_2\\ \text{equitable 2-colorings of }\mathcal{H}} } \Prob{ \text{Both } c_1 \text{ and } c_2 \text{ are proper} }
        \]
        \begin{equation}
        \label{eq:EX2}
            = \sum_{s=0}^{n/2} \binom{n}{n/2} \binom{n/2}{s}^2 \prod_k \left( 1 - \frac{4\binom{n/2}{k} -2\binom{s}{k} - 2\binom{n/2 - s}{k} }{ \binom{n}{k} } \right)^{M_k},
        \end{equation}
        where $s$ is the number of vertices assigned \emph{RED} in both colorings. By equitability, $s$ is also the number of vertices colored \emph{BLUE} in both colorings. The first coloring can be chosen in $\binom{n}{n/2}$ ways, then the choice of $s$ \emph{RED} and~$s$ \emph{BLUE} vertices determines the second coloring. Given a coloring, there are $2\binom{n/2}{k}$ ways for a $k$-edge to be monochromatic. The rest follows from subtracting double-counted configurations.

        Combining (\ref{eq:EX}) and (\ref{eq:EX2}), we get
        \[
            \frac{ \EV{X^2} }{ \EV{X}^2 }
            = \sum_{s=0}^{n/2} \frac{1}{ \binom{n}{n/2}^2 }  \binom{n}{n/2} \binom{n/2}{s}^2 \prod_k \left( \frac{ 1 - \frac{4\binom{n/2}{k} -2\binom{s}{k} - 2\binom{n/2 - s}{k} }{ \binom{n}{k} } }{ \left( 1 - \frac{2 \binom{n/2}{k}}{ \binom{n}{k} } \right)^2 } \right)^{M_k}
        \]

        After minor simplifications, we can split this expression into two parts. Let
        \[
            A(s) := \frac{1}{ \binom{n}{n/2} } \binom{n/2}{s}^2 
            \text{ \quad and \quad }
            B(s) := \prod_k \left( \frac{ 1 - \frac{4\binom{n/2}{k} -2\binom{s}{k} - 2\binom{n/2 - s}{k} }{ \binom{n}{k} } }{ \left( 1 - \frac{2 \binom{n/2}{k}}{ \binom{n}{k} } \right)^2 } \right)^{M_k}
        \]
        Observe that both $A$ and $B$ are symmetric with respect to $n/4$. Furthermore, $A$ is increasing and $B$ is decreasing for $s \in \{ 0, 1, \ldots, n/4 \}$. Therefore, it suffices to show
        \begin{equation}
        \label{eq:sumAB}
            \sum_{s=0}^{n/4} A(s)B(s) = \Oh(1)
        \end{equation}

        We will bound the sum above in two separate intervals. First, recall the following approximation of $\binom{n}{pn}$
        \begin{equation}
        \label{eq:approxBinom}
            \binom{n}{pn} = \frac{1}{ \sqrt{2\pi n \cdot p(1-p)} } \cdot 2^{n \cdot H(p)} \cdot \left( 1 + \oh(1) \right),
        \end{equation}
        where $H(p) := -p\log(p) - (1-p)\log(1-p)$ is the binary entropy function. As the approximation follows directly from Stirling's formula, it is valid as long as $pn = \omega(1)$ and $(1-p)n = \omega(1)$. 

        \subsubsection{Initial interval}
            \hfill \\
            First, we analyze the sum (\ref{eq:sumAB}) for $s \in \{0, 1, \ldots, \delta n \}$ for small $\delta > 0$ to be determined later. Note that when $s = 0$ the corresponding summand in (\ref{eq:sumAB}) is exactly $\EV{X}^{-1}$. Since $B(s)$ is decreasing we can bound the sum in the \emph{initial interval} as follows
            \begin{equation}
                \label{eq:initialSum}
                \sum_{s=0}^{ \delta n } A(s) B(s) 
                \leq \sum_{s=0}^{ \delta n } A(s) B(0)
                = \EV{X}^{-1} \sum_{s=0}^{ \delta n } \binom{n/2}{s}^2
                < n \cdot \EV{X}^{-1} \cdot \binom{n/2}{ \delta n }^2
                \leq n \cdot \EV{X}^{-1} \cdot 2^{n \cdot H(2 \delta)}
            \end{equation}

            Since $\lim_{\delta \rightarrow 0^+} H(2\delta) = 0$, to show that the contribution from the initial interval is $\oh(1)$, it is sufficient to show that $\EV{X}$ is exponentially large in $n$.

            \begin{restatable}[{$\EV{X}$} is exponentially large]{lem}{exponentiallyManyColorings} 
                \label{lem:exponentiallyManyColorings}
                \hfill \\
                In our regime $\EV{X}$ is exponentially large in $n$.
            \end{restatable}
            \begin{proof} 
                \hfill \\
                This proof is fairly similar to the proof of Theorem \ref{thm:nonUniformUpperBound}. We have
                \begin{align*}
                    \EV{X} & = \binom{n}{n/2} \prod_k \left( 1 - \frac{2 \binom{n/2}{k}}{ \binom{n}{k} } \right)^{M_k}
                    = \binom{n}{n/2} \prod_k \left( 1 - 2\frac{ (n/2)^{\underline{k} } } {n^{\underline{k}}} \right)^{M_k}
                    \\&= \binom{n}{n/2} \prod_k \left( \left( 1 - 2^{1-k} \right) \left( 1 + \Oh\left( n^{-1} \right) \right) \right)^{M_k}
                   \\& = \left( 1 + \oh(1) \right) \cdot \frac{2^n}{\sqrt{ \frac{1}{2}\pi n }} \cdot \prod_k \left( \left( \left( 1 - 2^{1-k} \right) \left( 1 + \Oh\left( n^{-1} \right) \right) \right)^{\lambda_k \cdot 2^k} \right)^n
                    \\&= \frac{\Theta(1)}{\sqrt{n}} \cdot 2^n \cdot \prod_k \left( \left( 1 - 2^{1-k} \right)^{\lambda_k \cdot 2^k} \right)^n
                \end{align*}
 
                Thus,
                \begin{align*}
                    \frac{1}{n}\ln \EV{X}
                    &= \ln 2 + \sum_k \lambda_k 2^k \ln\left( 1 - 2^{1-k} \right) + \oh(1) \\
                    &= \ln 2 - \sum_k \lambda_k \cdot \left( 2 + 2^{1-k} + \Oh(2^{-2k}) \right) + \oh(1) \\
                    &= \ln 2 - 2\sum_k \frac{\lambda_k}{1 - 2^{-k}} \left( 1 + \Oh\left( 2^{-2k} \right) \right) + \oh(1)
                \end{align*}
                By choosing $k_0 = k_0(\varepsilon)$ sufficiently large, we may ensure that $1 + \Oh\left( 2^{-2k} \right) \leq 1 + \varepsilon$. This yields
                \begin{align*}
                    \frac{1}{n} \ln \EV{X}
                    &\geq \ln 2 - 2\left(1 + \varepsilon\right)\sum_k \frac{\lambda_k}{1 - 2^{-k}} + \oh(1) \\
                    &> \ln 2 - \left(1 - \varepsilon^2\right) \ln 2 + \oh(1)
                \end{align*}
                Therefore, $\EV{X} \geq \exp \left( \frac{\varepsilon^2 \ln 2}{2} n \right)$.
            \end{proof}

            Let $\gamma :=  \frac{\varepsilon^2 \ln 2}{2}$. Choose $\delta_I = \delta_I(\varepsilon)$ sufficiently small that $H(2\delta_I) \ln 2 < \gamma$. Then 
            \[
                \sum_{s = 0}^{\delta_I n} A(s)B(s) = \oh(1)
            \]
        
        \subsubsection{Middle interval}
            Let $s_I := \lceil \delta_I n \rceil$. We have successfully bounded the contribution to our sum from the initial $s_I$ values of $s$ and we want to do the same for $s \in \{s_I, \ldots, \frac{n}{4}\}$. Note that such $s$' are bounded away from both $0$ and $n/2$, hence we can use Stirling's approximations of binomials (\ref{eq:approxBinom}).

            For $s\in \{s_I, \ldots, \frac{n}{4}\}$ we have
            \begin{align*}
                \label{FUN:A}
                A(s) &= \frac{1}{ \binom{n}{n/2} } \binom{n/2}{s}^2
                = \frac{ \sqrt{\frac{1}{2} \pi n} }{ 2^n } \cdot \frac{1}{ \pi n \cdot \frac{2s}{n} \left( 1 - \frac{2s}{n} \right) } \cdot 2^{H\left( 2s / n \right) n} \cdot \left( 1 + \oh(1) \right) \\
                &= \frac{ \left( 1 + \oh(1) \right) }{ \sqrt{2\pi n} \cdot \frac{2s}{n} \left(1 - \frac{2s}{n}\right) } \left( 2^{H\left( 2s / n \right) - 1} \right)^n
            \end{align*}
            and
            \begin{align*}
                B(s) &= \prod_k \left( \frac{ 1 - \frac{4\binom{n/2}{k} -2\binom{s}{k} - 2\binom{n/2 - s}{k} }{ \binom{n}{k} } }{ \left( 1 - \frac{2 \binom{n/2}{k}}{ \binom{n}{k} } \right)^2 } \right)^{M_k} \\
                &= \prod_k \left( \frac{ 1 - \frac{4(n/2)^k - 2s^k - 2(n/2 - s)^k}{n^k} }{ \left( 1 - \frac{2(n/2)^k}{n^k} \right)^2 } \cdot \left( 1 + \Oh\left( n^{-1} \right) \right) \right)^{\lambda_k \cdot 2^k \cdot n} \\
                &= \prod_k \left( \frac{ 1 - 4(1/2)^k + 2(s/n)^k + 2(1/2 - s/n)^k }{ \left( 1 - 2(1/2)^k \right)^2 } \right)^{\lambda_k \cdot 2^k \cdot n} \cdot \Oh(1)
            \end{align*}

            Based on these approximations, we define two real functions on the interval $\left[ 0, 1/2 \right]$.
            \[
                a(x) := 2^{H(2x) - 1} 
                \text{\quad and \quad} 
                b(x) := \prod_k \left( 1 + \frac{ -4(1/2)^{2k} + 2x^k + 2(1/2 - x)^k }{ \left( 1 - 2(1/2)^k \right)^2 } \right)^{\lambda_k \cdot 2^k}
            \]

            Note that $A(s) = \Oh\left( \frac{1}{\sqrt{n}} \cdot a(s/n)^n \right)$ and $B(s) = \Oh\left( b(s/n)^n \right)$. Together with the bound on the initial interval we get
            \[
                \frac{\EV{X^2}}{\EV{X}^2} = \Oh\left( \sum_{s = s_I}^{n/4} n^{-1/2} \cdot \left( a(s/n) \cdot b(s/n) \right)^n \right) + \oh(1)
            \]
            Therefore, from now on we focus on $a$ and $b$.

            Using $1+x \leq e^x$
            \[
                b(x) \leq \prod_k \exp\left( \lambda_k \cdot 2^k \cdot \frac{ -4(1/2)^{2k} + 2x^k + 2(1/2 - x)^k }{ \left( 1 - 2(1/2)^k \right)^2 } \right)
                = \exp\left( \sum_k 2\lambda_k \cdot \frac{ -2^{1-k} + (2x)^k + (1 - 2x)^k }{ \left( 1 - 2^{1-k} \right)^2 } \right)
            \]

            Let
            \[
                F(x) := \left( H(x) - 1 \right) \ln 2 + \sum_k 2\lambda_k \frac{ -2^{1-k} + x^k + (1-x)^k }{ \left( 1 - 2^{1-k} \right)^2 }
            \]
            Consequently
            \[
                \frac{\EV{ X^2 }}{\EV{ X }^2} 
                = \Oh\left( n^{-1/2} \sum_{s = s_I}^{n/4} \exp\left( n \cdot F \left( \frac{2s}{n} \right) \right) \right) + \oh(1)
            \]

            \begin{restatable}[$F$ is uniformly strictly concave]{lem}{uniformConcave} 
                \label{lem:uniformConcave}
                \hfill \\
                For sufficiently large $k_0$, there exists $c > 0$ such that $F''(x) \leq -c$ for every $x \in [2\delta_I, 1/2]$.
            \end{restatable}
            \begin{proof}
                \hfill \\
                We prove this lemma by analyzing $F$'s second derivative
                \[
                    F''(x) = - \frac{1}{x(1-x)} + \sum_k 2\lambda_k \frac{k(k-1)\left( x^{k-2} + (1-x)^{k-2} \right)}{\left( 1 - 2^{1-k} \right)^2}
                \]
                Note that
                \[
                    -\frac{1}{x(1-x)} \leq -4
                    \text{\quad and, for } k \geq 2 \text{, }
                    \frac{1}{\left( 1 - 2^{1-k} \right)^2} \leq 4
                \]
                Moreover, for $x \in [2\delta_I, 1/2]$
                \[
                    x^{k-2} + (1-x)^{k-2} \leq 2^{2-k} + (1 - 2\delta_I)^{k-2}
                \]
                Therefore
                \begin{align*}
                    \sum_k 2\lambda_k \frac{k(k-1)\left( x^{k-2} + (1-x)^{k-2} \right)}{\left( 1 - 2^{1-k} \right)^2}
                    &\leq 8 \sum_k \lambda_k \cdot \sup_{k \geq k_0} k(k-1)\left( x^{k-2} + (1-x)^{k-2} \right) \\
                    &\leq 8 \sum_k \lambda_k \cdot \sup_{k \geq k_0} k(k-1)\left( 2^{2-k} + (1 - 2\delta_I)^{k-2} \right) \\
                    &\leq 4\ln 2 \sup_{k \geq k_0 }  k(k-1)\left( 2^{2-k} + (1 - 2\delta_I)^{k-2} \right)
                \end{align*}
                The supremum tends to $0$ as $k_0 \to \infty$. Therefore, for sufficiently large $k_0$, there exists some $c > 0$ such that $F''(x) < -c$ throughout the interval.
            \end{proof}

            Note that $F(x) = F(1 - x)$. Therefore, $F'(1/2) = 0$. Direct substitution gives $F(1/2) = 0$ as well. Hence, by Taylor's theorem, for every $x \in [2\delta_I, 1/2]$ there exists some $\xi$ such that
            \begin{align*}
                F(x) &= F(1/2) + F'(1/2)(x - 1/2) + \frac{1}{2}F''(\xi)(x - 1/2)^2 \\
                &= \frac{1}{2}F''(\xi)(x - 1/2)^2 \\
                &\leq -\frac{c}{2}(1/2 - x)^2
            \end{align*}

            Substituting $x = 2s / n$, we get $F(2s / n) \leq -\frac{c}{2}\left( 1/2 - 2s / n \right)^2$. Therefore 
            \begin{align*}
                n^{-1/2} \sum_{s = s_I}^{n/4} \exp \left( n \cdot F(2s / n) \right) 
                &\leq n^{-1/2} \sum_{s = s_I}^{n/4} \exp \left( -\frac{2c}{n} (n/4 - s)^2 \right) \\
                &\leq n^{-1/2} \sum_{d = 0}^{\infty} \exp \left( -\frac{2c}{n} d^2 \right) \\
                &\leq n^{-1/2} \left( 1 + \int_{0}^{\infty} \exp \left( -\frac{2c}{n} t^2 \right) dt \right) \\
                &= n^{-1/2} + \sqrt{\pi / 8c} = \Oh(1)
            \end{align*}
            Combining this with the bound on the initial interval, we get
            \[
                \frac{\EV{X^2}}{\EV{X}^2} = \Oh(1)
            \]
            and the Paley-Zygmund inequality finishes the proof.
    \end{proof}

\end{document}